\documentclass{article}

\usepackage{amssymb}
\usepackage{epstopdf}
\usepackage{mathrsfs}
\usepackage{wasysym}
\usepackage{xspace}
\usepackage{pxfonts}

\newtheorem{theorem}{Theorem}

\newcommand{\ltl}{LTL\xspace}
\newcommand{\etl}{ETL\xspace}
\newcommand{\qltl}{QLTL\xspace}
\newcommand{\rltl}{RLTL\xspace}
\newcommand{\mutl}{$\mu$TL\xspace}
\newcommand{\buechi}{B\"uchi\xspace}

\newcommand{\op}[1]{\mathsf{#1}}
\newcommand{\sig}[1]{\mathcal{#1}}
\newcommand{\lang}{\mathscr{L}}
\newenvironment{proof}{\noindent \textsc{Proof}.}{\hfill$\square$}

\begin{document}

\title{A Short Note on Infinite Union/Intersection of Omega Regular Languages
\thanks{The author would thank Normann Decker, Daniel Thoma, Fu Song, Lei Song \& Lijun Zhang for
the fruitful discussions on this problem.}}
\author{Wanwei Liu \\
{\small School of Computer Science} \\
{\small National University of Defense Technology} \\
{\small Changsha, China, 410073}}
\date{}
\maketitle


\section{A Basic Observation}
\label{sec:observation}

The most impressive non-star-free property,
first pointed out by Wolper \cite{Wol83}, ``$p$ holds at every even moment''
(we in what follows refer to it as $\op{P}(2)$)
cannot be expressed by any \ltl \cite{Pnu77} formula.
As a consequence, numerous extensions or \ltl have been presented,
such as \etl \cite{VW94}, \qltl \cite{SVW87}, \rltl \cite{LS07},
linear-time \mutl \cite{BB87} etc,
and all of them are known to be as expressive as (nondeterministic)
\buechi automata \cite{Buc62}, alternatively, $\omega$-regular languages.

Indeed, $\lang(\op{P}(2))=\bigcap_{k\in\mathbb{N}}\lang(\op{X}^{2k}p)$,
and this indicates that star-free languages are not closed under infinite intersection.
It naturally enlightens us to make one step ahead, and now the question of
interest is:
\begin{quote}
  ``Are $\omega$-regular languages closed under infinite union/intersection?''
\end{quote}
For this, we just consider the language $\bigcup_{k\in\mathbb{N}}\lang(\op{P}(k))$,
which consists of all $\omega$-words along which $p$ holds periodically.

\begin{theorem}
\label{thm:non-close}
  The language $\bigcup_{k\in\mathbb{N}}\lang(\op{P}(k))$ is not $\omega$-regular.
\end{theorem}

\begin{proof}
Assume by contradiction that this
language is $\omega$-regular, then there is a nondeterministic \buechi
automaton $\mathcal{A}$  precisely recognizing it.
Namely, each $\omega$-word being of the form $(p;\neg p^k)^\omega$
must belong to $\lang(\mathcal{A})$, where $k\in\mathbb{N}$.
\par W.o.l.g., suppose that $\mathcal{A}$ has $n$ states,
and let us fix some $m>n+1$, then $\mathcal{A}$ has an accepting run over the
word $w = (p;\neg p^m)^\omega$, say $\sigma(0),\sigma(1),\ldots$
\par
From the Pumping lemma, for each $t\in\mathbb{N}$, there exists a pair $(i,j)$, s.t. $0<i<j<m+1$,
and $\sigma(t\times (m+1)+i)=\sigma(t\times (m+1)+j)$.
This implies that, for each $\ell$, the word
$$
[(p\cdot\neg p^{i-1})\cdot
(\neg p)^{(j-i)\times\ell}\cdot(\neg p)^{m-j}]\cdot(p\cdot\neg p^m)^\omega$$
is also acceptable by $\mathcal{A}$ with the run
$$\sigma(0),\sigma(1),\ldots,\sigma(i),
\underbrace{\sigma(i+1),\ldots,\sigma(j)=\sigma(i),
\ldots,
\sigma(i+1),\ldots,\sigma(j)}_{\ell \textrm{ times}}
,\sigma(j+1),\ldots$$
which is definitely accepting.
\par
Likewise and stepwise, we may  obtain a sequence of omega words as followings:
\begin{itemize}
\item[-] $w_0=(p\cdot\neg p^m)^\omega$.
\item[-] $w_1=(p\cdot \neg p^{L_1})\cdot(p\cdot\neg p^m)^\omega$.
\item[-] $w_2=(p\cdot \neg p^{L_1})\cdot(p\cdot \neg p^{L_2})\cdot(p\cdot\neg p^m)^\omega$.
\item[-] $\cdots$
\end{itemize}
where $L_1<L_2<L_3<\cdots$
\par
Then, $\mathcal{A}$ also has an accepting run on the limit of the sequence
$$w_\infty=(p\cdot\neg p^{L_1})\cdot(p\cdot\neg p^{L_2})\cdot(p\cdot\neg p^{L_3})\cdot\ldots.$$
thus we can conclude that $w_\infty\in\lang(\mathcal{A})$.
\par
However, $w_\infty$ could not have a ``period'' on $p$ |
because, for every number $k$,
there must exist some $L_c>k$ | this implies  that
the distance between two adjacent occurrences of $p$
will be larger than $k$ in the future.
\par
Thus, we have got a contradiction\footnote{Remind that an $\omega$-word
$w\in\lang(\mathcal{A})$ iff $p$ rises periodically along $w$.},
and it lies from the assumption that $\bigcup_{k\in\mathbb{N}}\lang(\op{P}(k))$
is regular.
\end{proof}

Observe that each $\lang(\op{P}(k))$ is regular, but it is not the case for the union of all such languages,
and hence regular languages are not closed under infinite union and/or intersection.

\section{Adding Step Variables and Quantifiers?}
\label{sec:quantifier}

As a possible solution of expressing the aforementioned property, one may orthogonally add
\emph{step variables} as well \emph{step quantifiers} in temporal logics involving next operator ($\op{X}$).
In a (closed) formula, a step variable is introduced by a quantifier and appears associated with a next operator.

Syntax and semantics of such kind of extensions can be naturally and
succinctly obtained w.r.t. the underlying logics.
As an example, let $\sig{P}$ be the set of propositions, and
let $\sig{K}$ be the collection of all step variables,
formulae (in PNF, ranging over $f$, $g$, etc) of \ltl with such features can be
described by the following abstract grammar.
$$f \Coloneqq \top\mid \bot \mid p \mid \neg p \mid f\wedge f \mid f\vee f
    \mid \op{X} f \mid \op{X}^k f \mid f\op{U}f \mid f\op{R}f \mid \exists k. f \mid \forall k. f$$
where $p\in\sig{P}$ and $k\in\sig{K}$.
The \emph{satisfaction} ($\models$) of a formula can be defined w.r.t. an $\omega$-word $\pi\in(2^\sig{P})^\omega$,
a position $i\in\mathbb{N}$ and a valuation $v :\sig{K}\to\mathbb{N}$.
Most cases are defined as routine, and
\begin{itemize}
  \item $\pi,i,v\models\op{X}^k f$ iff $\pi,i+v(k),v\models f$.
  \item $\pi,i,v\models\exists k. f$ (resp. $\pi,i,v\models\forall k.f$) iff $\pi,i,v[k/n]\models f$ for some
    (resp. for every) $n\in\mathbb{N}$.
\end{itemize}

One can, of course, choose linear \mutl as the base logic\footnote{
Note that to define semantics of such an extension, another
valuation from $\sig{Z}$ to $2^\mathbb{N}$ is also required,
where $\sig{Z}$ is the set consisting of all predicate variables.}.
Henceforth, we obtain formulae like
$$\forall k. \op{X}^k\op{X}^k p
\qquad \textrm{and}\qquad
\exists k. \nu Z.(p\wedge\op{X}^k Z).$$
Actually, the former is just $\op{P}(2)$, and the latter precisely describes
the property ``$p$ occurs periodically'' | which is not an $\omega$-regular property.

\section{On Decidability of Such Extensions}
\label{sec:decision}

Although adding step variables and step quantifiers to logics seems to be a natural and succinct
solution, we in this section reveal an inadequate feature of this mechanism
| the \textsc{satisfiability} problem, even if for the ``core fragment'' given by
$$f\Coloneqq\bot\mid\top\mid p\mid\neg p\mid f\wedge f\mid f\vee f \mid \op{X}f
\mid\op{X}^k f\mid\exists k.f \mid \forall k. f$$
is not decidable!

But before giving the proof, let us define some syntactic sugars:
\begin{itemize}
  \item We respectively abbreviate $\underbrace{\op{X}\ldots\op{X}f}_{n\textrm{ times}}$
  and $\underbrace{\op{X}^k\ldots\op{X}^kf}_{n\textrm{ times}}$ as $\op{X}^n f$ and
  $\op{X}^{n\cdot k}f$, where $n\in\mathbb{N}$ and $k\in\sig{K}$.
  \item We sometimes directly write $\op{X}^{t_1}\op{X}^{t_2}f$ as $\op{X}^{t_1+t_2}f$, provided
  that each $t_i$ is of the form $(\sum_{j}n_j\cdot k_j)+n$, where $n$ and each $n_j$ are natural numbers
  and each each $k_i\in\sig{K}$.
\end{itemize}

Note that in this setting, both the addition ($+$) and the multiplication ($\cdot$)
are communicative and associative. Meanwhile, ``$\cdot$'' is distributive w.r.t. ``$+$'',
namely, $t_1\cdot t_2 + t_1\cdot t_3$ can be rewritten as $t_1\cdot(t_2+t_3)$.

Moreover, for convenience, when $f$ is a formula involving no free variable, we directly
write $\pi,i,v\models f$ as $\pi,i\models f$. Further, we write $\pi\models f$ (resp. $\pi,v\models f$)
in place of $\pi,i\models f$ (resp. $\pi,i,v\models f$) in the case of $i=0$.

\begin{theorem}
  \label{thm:undecidability}
  The \textsc{satisfiability} problem of the core logic is not decidable.
\end{theorem}

\begin{proof}
  The main observation is that ``each formula of Peano arithmetic\footnote
  {Peano arithmetic is just a fragment of first order logic, with the signature consisting of naturals,
  the function $+$, $\times$ (respectively be interpreted as addition and
  multiplication), and the predicate $<$ (whose canonical interpretation is ``less than'').
  cf. \cite{Pea89}.} has a peer expression in
  this fragment, and they are of the same satisfiability''.
  \par
  To show this, we need to build the following predicates:
  \begin{enumerate}
    \item Fix a proposition $p\in\sig{P}$, and let
	 $L_p \triangleq\forall k_1.\forall k_2. \neg\forall k_3.
	(\mathsf{X}^{k_1+k_3}p \leftrightarrow\mathsf{X}^{k_1+k_2+k_3+1} p)$.
    \par	
	Actually, $L_p$ just depicts the ``\emph{non-shifting property}'' of $p$. i.e.,
	if $\pi\models L_p$ then
	for each $i,j\in\mathbb{N}$ with $i< j$, there is some $t$ having: either
	``$\pi,i+t\models p$ and $\pi,j+t\not\models p$'' or
	``$\pi,i+t\not\models p$ and $\pi,j+t\models p$''.
	(Just view $k_1$ as $i$ and view $k_1+k_2+1$ as $j$.)
    \item Let $L_=(k_1,k_2) \triangleq L_p\wedge
	\forall k.(\mathsf{X}^{k_1+k}p\leftrightarrow\mathsf{X}^{k_2+k}p)$.
    \par
	Hence, $\pi,i,v\models L_=(k_1,k_2)$ iff $v(k_1)=v(k_2)$.
    \item Let $L_<(k_1,k_2)\triangleq\exists k. L_=(k_1+k+1,k_2)$.
    \par
	Clearly, $\pi,i,v\models L_<(k_1,k_2)$ iff $v(k_1)<v(k_2)$.
    \item Subsequently, we use $L_+(k_1,k_2,k_3)$ to denote $L_=(k_1+k_2,k_3)$.
    \par
    According to the definition, $\pi,i,v\models L_+(k_1,k_2,k_3)$ iff
    $v(k_1+k_2)=v(k_3)$.
    \item Now, let us fix another proposition $q\in\sig{P}$ and define
	$$\begin{array}{rcl}
	L_q  & \triangleq  & q~\wedge~\mathsf{X}q~\wedge~\forall k_1.\exists k_2.\mathsf{X}^{k_1+k_2} q~\wedge \\
		&		&	\forall k_1.\forall k_2. \forall k_3. (\mathsf{X}^{k_1}q
		\wedge\mathsf{X}^{k_2} q\wedge\mathsf{X}^{k_3}q \\
		&		& \wedge L_<(k_1,k_2)\wedge L_<(k_2,k_3) \\
		&		& \wedge \forall k_4.(L_<(k_1,k_4)\wedge L_<(k_4,k_2)\vee L_<(k_2,k_4)\wedge L_<(k_4,k_3)\rightarrow\neg\mathsf{X}^{k_4}q)\\
		&		& \rightarrow\exists k_5.\exists k_6. (L_+(k_5,k_1,k_2)\wedge L_+(k_6,k_2,k_3)\wedge L_+(2,k_5,k_6)))
	\end{array}$$
    \par
	We may assert that $\pi,i\models L_q$ iff $i$ is a complete square number (i.e., $i=j^2$ for some $j$).
	Let us explain : The first line indicates that $q$ holds infinitely often, and it holds at the positions of $0$ and $1$.
	For every three adjacent positions $k_1$, $k_2$, $k_3$ at which $q$ holds (hence, $q$ does not hold between
	$k_1$ and $k_2$, nor between $k_2$ and $k_3$), we have $|k_3-k_2| = 2+ |k_2-k_1|$.
	Inductively, we can show that $q$ becomes true only at $0$, $1$, $4$, \ldots, $(n-1)^2$, $n^2$, $(n+1)^2$, \ldots.
    (The encoding of $L_q$ is enlightened by \cite{Sch10}.)
    \item We let $$L_2(k_1,k_2)\triangleq L_q\wedge\mathsf{X}^{k_2}q\wedge\mathsf{X}^{k_2+2\cdot k_1+1}q
	\wedge\neg\exists k_3.(L_<(k_2,k_3)\wedge L_<(k_3, 2\cdot k_1+k_2+1)\wedge\mathsf{X}^{k_3}q)$$
	then we have that $\pi,v\models L_2(k_1,k_2)$ iff $v(k_2)=(v(k_1))^2$.
    \item As the last step, we  define that
	$$\begin{array}{rcl}
	L_\times(k_1,k_2,k_3) & \triangleq & \exists k_4.\exists k_5.\exists k_6. ( L_2(k_1,k_4)\wedge L_2(k_2,k_5)\\
			&	&\wedge L_2(k_1+k_2,k_6)\wedge L_=(k_4+k_5+2\cdot k_3, k_6))
	\end{array}$$
	Then in the case of $\pi,i,v\models L_\times(k_1,k_2,k_3)$,  we
	may get the following constraints:
	$$\left\{\begin{array}{rcl}
	v(k_4) & = & (v(k_1))^2 \\
	v(k_5) & = & (v(k_2))^2 \\
	v(k_6) & = & (v(k_1)+v(k_2))^2 \\
	v(k_6) &= & v(k_4)+v(k_5)+2\times v(k_3)
	\end{array}\right.$$
	and we subsequently have $v(k_3)=v(k_1)\times v(k_2)$.
  \end{enumerate}
  Now, we can see that ``addition'', ``multiplication'', and the ``less than'' relation over natural numbers can
  be encoded in terms of the core logic. Since quantifiers are also involved here, then the \textsc{satisfiability}
  problem of Peano arithmetic can be reduced to that of the core logic | the former is known to be undecidable
 (cf. \cite{God31Deu,Chu36}).
\end{proof}

\section{Further Discussions}
\label{sec:discussion}

As we have seen, to gain the expressiveness of infinite union/intersection of
(a family of) regular languages, an admissible approach is to cooperate with
step variables and quantifiers in the logic | however, it suffers from the
undecidability of \textsc{satisfiability}.

To tackle this, we need to investigate new mechanisms | it should both
enhance the expressiveness and keep the logic decidable.
So far, we are not aware of it, and it seems that employing more powerful existing
automata, say pushdown automata, is also not feasible.


\end{document}